\documentclass[conference]{IEEEtran}

\usepackage[margin=2cm]{geometry}
\usepackage{amssymb,amsthm,amsmath}
\usepackage[utf8]{inputenc}
\usepackage[T1]{fontenc}
\usepackage{cite}
\usepackage{ifthen}
\usepackage{bm}
\usepackage[tracking=true]{microtype}
\usepackage[english]{babel}
\usepackage{latexsym}
\usepackage{graphicx}
\usepackage{subfig}
\usepackage{caption}
\usepackage{pgf,tikz}
	\usetikzlibrary{decorations.pathreplacing,patterns,arrows}
\usepackage{dsfont}
\usepackage{bm}
\usepackage{placeins}

\newcommand{\tf}{\widetilde{f}}

\newcommand{\cR}{\mathcal{R}}

\newcommand{\Rtt}{\cR^{\mathrm{true}}}
\newcommand{\Rbt}{\cR^{\mathrm{c.i.}}}
\newcommand{\Rbg}{\cR^{\star}}
\newcommand{\Rbu}{\cR^{\mathrm{up}}}

\newcommand{\R}{\mathbb{R}}
\newcommand{\N}{\mathcal{N}}

\renewcommand{\P}{\mathbb{P}_{\mu}}

\newcommand{\C}{\mathcal{L}}
\newcommand{\E}{\mathbb{E}}
\newcommand{\Var}{\mathrm{Var}}

\newtheorem{thm}{Theorem}[section]
\newtheorem{lem}{Lemma}[section]
\newtheorem{prop}{Proposition}[section]

\begin{document}

\title{Line failure probability bounds for power grids}

\author{\IEEEauthorblockN{Tommaso Nesti}
\IEEEauthorblockA{CWI, Amsterdam\\
\texttt{nesti@cwi.nl}}
\and
\IEEEauthorblockN{Alessandro Zocca}
\IEEEauthorblockA{CWI, Amsterdam\\
\texttt{zocca@cwi.nl}}
\and
\IEEEauthorblockN{Bert Zwart}
\IEEEauthorblockA{CWI, Amsterdam, and\\
Eindhoven University of Technology\\
\texttt{bert.zwart@cwi.nl}}}

\maketitle

\begin{abstract}
We develop upper bounds for line failure probabilities in power grids, under the DC approximation and assuming Gaussian noise for the power injections. Our upper bounds are explicit, and lead to characterization of safe operational capacity regions that are convex and polyhedral, making our tools compatible with existing planning methods. Our probabilistic bounds are derived through the use of powerful concentration inequalities.
\end{abstract}

\IEEEpeerreviewmaketitle
\section{Introduction}
\label{sec1}
Electrical power grids are expected to be reliable at all times. The rise of intermittent renewable generation is making this expectation challenging to live up to.  Power imbalances caused by generation intermittency may cause grid stability constraints to be violated:~80\% of the bottlenecks in the European high-voltage grid was already caused by renewables in 2015~\cite{Yang2015b}. A well-controlled power grid matches supply and demand, ensuring that line constraints are not violated. System operators achieve this by making periodic control actions that adapt the operating point of the grid in response to changing conditions~\cite{Ela2011}.

Due to the impact of renewables, a planning that accounts for worst-case behavior may lead to overly conservative solutions. A more realistic paradigm is to make a planning admissible when the probability that line power flows exceed a threshold is sufficiently small. This has motivated several recent works that attempt to evaluate line failure probabilities using rare event simulation techniques~\cite{Wadman2012,Wadman2013,Shortle2013}, as well as large deviations techniques~\cite{Nesti2016}. Simulation techniques can lead to accurate estimates, but may be too time-consuming to use as subroutine within an optimization package that has to determine a planning that is operational during the next 5 to 15 minutes, such as optimal power flow (OPF). Recent papers studying chance-constrained versions of OPF problems include~\cite{Bienstock2014,Summers2014}. Large deviations techniques are appealing, but rely on a scaling procedure, essentially assuming that the noise during the next planning period is small.

This article makes a new contribution to this emerging area by deriving approximations for line failure probabilities that are \textit{guaranteed to be conservative}. That is, keeping the approximation smaller ensures that reliability constraints are actually met. In addition, these new approximations are explicit enough to be used for optimization purposes on short time scales. In particular, we develop two such approximations in Section~\ref{sec3}. Both bounds lead to an approximation of the capacity region that is conservative, convex and polyhedral, making our results compatible with existing planning methods like OPF~\cite{Bienstock2014,Summers2014}.

This paper is organized as follows. In Section~\ref{sec2} we provide a detailed problem formulation. We model stochastic power injections into the network by means of Gaussian random variables, describe line power flows through the well-known DC approximation, and define the failure probabilities of interest. Our main results are two different upper bounds that we present in Section~\ref{sec3}. The first upper bound is explicit, while the second one is sharper and explicit up to a finite-step minimization procedure. These bounds are compared numerically with the exact safe capacity regions in Section~\ref{sec4}. Section~\ref{sec6} provides the proofs of the results in Section~\ref{sec3}. Concluding remarks are provided in Section~\ref{sec5}.


\section{Problem formulation}
\label{sec2}

\subsection{Network description and DC approximation}
\label{sub:network}
We model the power grid network as a connected graph $G=G(V,E)$, where $V$ denotes the set of \textit{buses} and $E$ the set of directed edges modeling the \textit{transmission lines}. $n=|V|$ is the number of buses and $m=|E|$ is the number of lines. $(i,j) \in E$ denotes the transmission line between buses $i$ and $j$ with \textit{susceptance} $\beta_{i,j}=\beta_{j,i}$. If there is no transmission line between $i$ and $j$ we set $\beta_{i,j}=\beta_{j,i}=0$.

As in~\cite{Cetinay2016,Zocca2016}, the network structure and susceptances are encoded in the \textit{weighted Laplacian matrix}
\[
	L_{i,j} :=\begin{cases}
		-\beta_{i,j}								&\text{if } i \neq j,\\
		\sum_{k\neq j} \beta_{i,k} 	&\text{if } i=j.
	\end{cases}
\]
Let $p\in\R^n$ denote the vector of (real) power injections, $\theta\in\R^n$ the vector of phase angles, and $\tf\in\R^m$ the vector of (real) power flow over the lines. We will use the convention that $p_i\ge 0$ ($p_i<0$) means that power is generated (consumed, respectively) at bus $i$.

We make use of the \textit{DC approximation}, which is commonly used in transmission system analysis~\cite{Purchala2005,Stott2009,Powell2004,Wood2012}. Thus, the real flow $\tf_{i,j}$ over line $(i,j)$ is related to the phase angle difference between buses $i$ and $j$ via the linear relation
\begin{equation}\label{eq:DC_scalar}
	\tf_{i,j}=\beta_{i,j}(\theta_i-\theta_j).
\end{equation}
We assume a balanced DC power flow, which means that the total net power injected in the network is zero, i.e.
\begin{equation}
\label{eq:zerosum}
	\mathbf{1}^Tp=0,
\end{equation}
where $\mathbf{1} \in \R^n$ is the vector with all entries equal to $1$. We enforce this constraint through the concept of \textit{slack bus}. 
Following the approach in~\cite{Cetinay2016}, and invoking assumption~\eqref{eq:zerosum}, the relation between  $\theta$ and $p$ can be written in matrix form as
\begin{equation}
\label{eq:dcapprox}
	\theta =L^+p,
\end{equation}
where $L^+$ is the \textit{Moore-Penrose} pseudo-inverse of the matrix $L$ and an average value of zero has been chosen as a reference for the node voltage phase angles. Choosing an arbitrary but fixed orientation of the transmission lines, the network structure is described by
the \textit{weighted edge-vertex incidence matrix} $B\in\R^{m\times n}$ whose components are
\[
 	B_{\ell, i}=\begin{cases}
    	\beta_{i,j}		&\text{if } \ell=(i,j),\\
    	-\beta_{i,j}		&\text{if } \ell=(j,i),\\
    	0    &\text{otherwise}.
    \end{cases}
 \]
Using such a matrix, we can rewrite identity~\eqref{eq:DC_scalar} as $\tf = B \theta$. Combining the latter equation and~\eqref{eq:dcapprox}, the line power flow $\tf$ can be written as a linear transformation of the power injections $p$, i.e.
\begin{equation}
\label{eq:tfBLp}
	\tf=BL^+ p.
\end{equation}
Transmission lines can fail due to overload. We say that a \textit{line overload} occurs in transmission line $\ell$ if $|\tf_\ell| > M_\ell$, where $M_\ell$ is the \textit{line capacity}. If this happens, the line may trip, causing a global redistribution of the line power flows which could trigger cascading failures and blackouts.
It is convenient to look at the \textit{normalized line power flow} vector $f \in \R^m$, defined component-wise as $f_\ell:= \tf_\ell / M_\ell$ for every $\ell=1,\dots,m$. The relation between line power flows and normalized power flows can be rewritten as $f = D \tf$, where $D  \in \R^{m \times m}$ is the diagonal matrix $D:=\mathrm{diag}(M_1^{-1}, \dots, M_m^{-1})$. In view of~\eqref{eq:tfBLp}, we have
\begin{equation}
\label{eq:fCp}
	f= C p,
\end{equation}
where $C:=DBL^+\in\R ^{m\times n}$. Henceforth, we refer to the normalized power flows simply as power flows, unless specified otherwise.


\subsection{Stochastic power injections and line power flows}
\label{sub:stochasticpower}

In this section we describe our model for the bus power injections. As our focus is on network reliability under uncertainty, we assume that each bus houses a \textit{stochastic power injection} or \textit{load}. This choice allows to model, for example, intermittent power generation by renewable sources or highly variable load.

In order to guarantee that network balance condition~\eqref{eq:zerosum} is satisfied even with stochastic inputs, we assume that bus $n$ is a \textit{slack bus}, which means that its power injection is chosen in such a way that the vector of actual power injections is a zero-sum vector as required in~\eqref{eq:zerosum}.

More specifically, we assume that the the vector of the first $n-1$ power injections $(p_1,\ldots,p_{n-1})$ follows a multivariate Gaussian distribution, with expected value $\mu \in \R^{n-1}$ and covariance matrix $\Sigma \in \R ^{(n-1) \times (n-1)}$. Since the covariance matrix $\Sigma$ is positive semi-definite, the matrix $\sqrt{\Sigma} \in \R^{ (n-1) \times (n-1)}$ is well defined via the Cholesky decomposition of $\Sigma$.  We are now able to formally define the vector $p$ of power injections as the $n$-dimensional random vector
\begin{equation}
\label{eq:powerinjections}
	p = S (\sqrt{\Sigma} X + \mu),
\end{equation}
where $X \sim \N_{n-1}(\bm{0}, I_{n-1})$ is a $(n-1)$-dimensional standard multivariate Gaussian random variable and $S$ is the matrix
\[
	S :=
	\left(
	\begin{array}{c}
		I_{n-1}\\
		- \mathbf{1}
	\end{array}
	\right) \in \R^{n \times (n-1)}.
\]
By construction we have $p=(p_1,\ldots,p_{n-1},- \sum_{i=1}^{n-1}p_i)$, so that~\eqref{eq:zerosum} is satisfied. Note that this formulation allows us to model \textit{deterministic power injections} as well, by means of choosing the corresponding variances and covariances equal to zero (or, from a practical standpoint, equal to very small positive numbers, so that the rank of~$\Sigma$ is not affected).

It is well known that an affine transformation of a multivariate Gaussian random variable is again a multivariate Gaussian random variable. Thus, identity~\eqref{eq:powerinjections} tells us that the power injections $p$ are indeed Gaussian, and hence, in view of~\eqref{eq:fCp}, so are the line power flows $f$. As it is convenient to look at the line power flows $f$ as an affine transformation of \textit{standard independent} Gaussian random variables, combining~\eqref{eq:fCp} and~\eqref{eq:powerinjections}, we can write
\begin{equation}
\label{eq:fXmu}
	f=V X + W \mu,
\end{equation}
where $V:= D C S \sqrt{\Sigma} \in \R^{m \times (n-1)}$ and $W:=D C S  \in \R^{m \times (n-1)}$. We denote by $\nu:=W \mu$ the vector of expected line power flows.

To summarize, we assume that the line power flows $f$ follow a multivariate Gaussian distribution $f\sim\N_{m}(\nu,VV^T)$, where the network topology and the correlation of the power injections are both encoded in the matrix $V$. Note in particular that $f_i \sim \N(\nu_i, \sigma_i^2)$, where the variance can be calculated as
\begin{equation}
\label{eq:defsigmai}
	\sigma_i^2:= \sum_{j=1}^n V_{i,j}^2.
\end{equation}

The main assumption behind our stochastic model is that the power injections are Gaussian. In~\cite[Section 1.5]{Bienstock2014} it is argued how this assumption, altough simplifying, is reasonable in order to model buses that house wind farms. Note that, compared to the power injections model in~\cite{Bienstock2014}, our formulation allows for general correlations between stochastic injections, as we do not impose any restrictions on the covariance matrix $\Sigma$. Section~\ref{sec5} contains a discussion to what extent our assumptions may be relaxed.

\subsection{Line failure probabilities}
\label{sub:problemstatement}
The main goal of the present paper is to understand how the probability of an overload violation depends on the parameters of the systems and characterize which average power injection vector $\mu$ will make such a probability smaller than a desired target.

In view of the definition of line overload given in Subsection~\ref{sub:network}, we define the \textit{line failure event} $\C$ as
$
	\C := \left \{ \exists \, \ell=1,\dots,m ~:~ |\tf_\ell| \geq M_\ell \right \}.
$
Leveraging the normalized line power flows that we introduced earlier, we can equivalently rewrite $\C$ as
\[
	\C= \left \{\max_i|f_i| \geq 1\right  \}.
\]
Given a power injection covariance matrix $\Sigma$, define the \textit{risk level} $r(\mu)$ associated with a power injection profile $\mu$ as
\[
	r(\mu):= \E \max_i |f_i|.
\]

Given a covariance matrix $\Sigma$, the risk level is a well-defined function $r: \R^{n-1} \to \R$ of the average injection vector $\mu$, since in view of~\eqref{eq:fXmu} we can rewrite $	r(\mu) = \E \max_i |V_i X + W_i \mu|$, where $V_i$ and $W_i$ denote the $i$-th row of the matrices $V$ and $W$, respectively.

We aim to characterize for a \textit{given} covariance matrix $\Sigma$ the average power injection vectors $\mu$ that make line failures a \textit{rare event}, say $\P(\C) \leq q$ for some very small threshold $q \in (0,1)$ to be set by the network operator (think of $q=10^{-5}$ or $q=10^{-6}$). In other words, given $q \in (0,1)$, we aim to determine the region $\Rtt_q \subset \R^{n-1}$ defined by
%
\[
	\Rtt_q:=\{\mu\in \R^{n-1} : \P(\C) \leq q\}.
\]
For every given $\mu \in \R^{n-1}$, calculating exactly the probability $\mathbb P_{\mu}(\C)$ means solving a high-dimensional integral that is also unavoidably error-prone, since the integrand becomes extremely small quickly (containing a multivariate Gaussian density). Hence, finding the region $\Rtt_q$ exactly is a very computationally expensive and error-prone task.

This is the main motivation of the present work, in which we develop analytic tools which are explicit enough to be useful for planning and control of power grids in the short-term. More specifically, in the next section we propose \textit{capacity regions} that can be calculated much faster and that can be used to approximate $\Rtt_q$.



\section{Main results}
\label{sec3}
This section is entirely devoted to the new three capacity regions $\Rbu_q, \Rbg_q$, and $\Rbt_q$ that we introduce to approximate $\Rtt_q$. We first introduce the probabilistic upper bounds on which our method is based in Subsection~\ref{sub31}, then formally define the regions $\Rbu_q, \Rbg_q$, and $\Rbt_q$ in Subsection~\ref{sub32} and lastly in Subsection~\ref{sub33} discuss the trade-offs between these different regions.

\subsection{Concentration inequalities}
\label{sub31}
Our methodology relies on a well-known \textit{concentration bound} for a function of Gaussian random variables.
Concentration bounds describe the likelihood of a function of many random variables to deviate from its expected value. In our context, we are interested in understanding how likely is the random variable $\max_i |f_i|$ to deviate from its expected value $r(\mu) =\E \max_i |f_i|$.

Many concentration bounds have been proved, see~\cite[Chapter 2]{Wainwright2015} for an overview. In our setting, we require Proposition~\ref{prop:borell}, which is presented and proved later in Section~\ref{sec6}. The next theorem presents an explicit upper bound for the line failure probability in terms of $r(\mu)=\E \max_i|f_i|$ and the variances $\sigma^2_1, \dots, \sigma^2_m$ of the line power flows that can be derived using the aforementioned concentration bound.

\begin{thm}[Upper bound for line failure probability]\label{thm:concentration}$ $\\
If $r(\mu)< 1$, then
\begin{equation}
\label{eq:bound_C}
	\P(\C) \leq \exp\Bigl(-\frac{(1-r(\mu))^2}{2 \max_i \sigma_i^2}\Bigr).
\end{equation}
\end{thm}
Note that $\E \max_i |f_i|=r(\mu)>1$ is definitely not a desirable operational regime for the power grid, since line failures are not rare events anymore.

\subsection{Capacity regions}
\label{sub32}
Given $q \in (0,1)$,  region $\Rbt_q$ is defined as the region that consists of all average power injection vectors $\mu$ such that the upper bound for $\P(\C)$ given by the concentration inequality~\eqref{eq:bound_C} is smaller than or equal to $q$, i.e.
\[
	\Rbt_q := \left \{\mu\in \R^{n-1} ~:~ \exp\Bigl(-\frac{(1-r(\mu))^2}{2 \max_i \sigma_i^2}\Bigr) \leq q \right \},
\]
which can be rewritten as
\[
	\Rbt_q  = \left \{\mu\in \R^{n-1} ~:~ r(\mu) \leq 1-\max_i \sigma_i \sqrt{2\log q^{-1}} \right \}.
\]
Unfortunately, the exact calculation of $r(\mu)$ is computationally expensive, for the same reasons outlined at the end of Section~\ref{sec2}. Furthermore, we want to have a better analytic understanding of the dependency of $r(\mu)$ on the power injection averages $\mu$, on the network topology and on the variances $\sigma_i$, something that is hard to obtain from purely numerical procedure. Aiming to overcome these issues, we propose an explicit upper bound for $r(\mu)$, namely
\begin{equation}
\label{eq:rup}
	r(\mu) \leq r^{\mathrm{up}}(\mu):= \max_i |\nu_i| + \max_i \sigma_i \sqrt{2\log(2m)},
\end{equation}
where we recall that $\nu=W\mu$ is the vector of average line power flows. The bound in~\eqref{eq:rup} is proven in Lemma~\ref{lem:upperbound_rmu} and can be used to obtain the following sub-region of $\Rbt_q$
\[
	\Rbu_q:=\left \{\mu\in \R^{n-1} ~:~ r^{\mathrm{up}}(\mu) \leq   1-\max_i \sigma_i \sqrt{2\log q^{-1} } \right \},
\]
which can be rewritten explicitly as
\begin{align*}
		\Rbu_q=\Bigl\{\mu & \in \R^{n-1} ~:~ \max_i |\nu_i| \leq \\
	  & \leq 1-\max_i \sigma_i (\sqrt{2\log q^{-1}}+\sqrt{2\log(2m)}\Bigr\}.
\end{align*}
In terms of $\mu$ we see that $\Rbu_q$ is an intersection of half-spaces, and so $\Rbu_q$ is convex and polyhedral.
A refinement of our analysis (see Lemma~\ref{lem:upperbound_rmu}) shows that is possible to obtain a sharper upper bound $r^\star(\mu)$ for $r(\mu)$,
\[
	r(\mu) \leq r^\star(\mu) \leq r^{\mathrm{up}}(\mu),
\]
which results in the following region
\[
	\Rbg_q:=\left \{\mu\in \R^{n-1} ~:~ r^\star(\mu) \leq 1-\max_i \sigma_i\sqrt{2\log q^{-1}} \right \}.
\]
Unfortunately there is no analytic expression for $r^\star(\mu)$, but in Section~\ref{sec6} we show that calculating $r^\star(\mu)$ requires only the evaluation of a function in a finite number of points, making it a numerically viable approach, and the resulting capacity region remains convex and polyhedral. Summarizing, we have
\begin{thm}[Inclusions among capacity regions]\label{thm:regions}
Given $q \in (0,1)$, if $r(\mu)<1$, then the following inclusions hold:
\begin{equation}\label{eq:inclusion}
	\Rbu_q \subseteq \Rbg_q \subseteq \Rbt_q \subseteq \Rtt_q.
\end{equation}
\end{thm}

\subsection{Discussion}
\label{sub33}
We can guarantee that a line overload is a sufficiently rare event by enforcing that the risk level $r(\mu)$ is at most $1-\max_{i}\sigma_i\sqrt{2\log(1/q)}$. This approach has the merit to provide a capacity region $\Rbt_q$ that can be expressed as a simple linear condition on the risk level $r(\mu)$, but has the drawback that it requires the computation of $r(\mu)$, a non-trivial task.

The smaller region $\Rbu_q$, although more conservative, is expressed in closed-form and, moreover, its dependency on the parameters $\nu,\sigma$ and $m$ is made explicit. In particular, the maximum standard deviation of the power flows, i.e.~$\max_i \sigma_i$ plays a big role in defining the capacity regions: indeed to larger values of $\max_i \sigma_i$ correspond smaller regions, which is intuitive since a bigger variance results in a higher probability of overload. In between the two regions $\Rbu_q$ and $\Rbt_q$ lies the intermediate region $\Rbg_q$, which is less conservative that $\Rbu_q$ and can be computed very efficiently, even if it cannot be expressed in closed-form (see Section~\ref{sec6} for more details).

Both regions $\Rbu_q$ and $\Rbg_q$ seem sufficiently explicit to be used as probabilistic constraints into chance-constrained versions of OPF problems, as studied in~\cite{Bienstock2014,Summers2014}.

\section{Numerical case studies}
\label{sec4}
To illustrate how the three new regions compare to $\Rtt_q$, we consider first a very simple network with a circuit topology, consisting of $3$ buses, all connected with each other by $3$ identical lines of unit susceptance and capacity $M=5$. We take the power injections in the non-slack nodes to be independent, zero-mean Gaussian random variables with variance $\epsilon=0.5$, which correspond to taking $\mu=(0,0)$ and $\Sigma= \epsilon I_2$. The corresponding four safe capacity regions with $q=10^{-3}$ are plotted in Figure~\ref{fig:3bus}.

We then plot in Figure~\ref{fig:14bus} the two-dimensional capacity regions $\Rbu_q$ and $\Rbg_q$ for the IEEE 14-bus test network (representing a portion of the American Electric Power System~\cite{Christie2006}) corresponding to bus $6$ and $9$. We replace the deterministic power injections with Gaussian random variables with average $\mu$ equal to the original deterministic values and variance $\epsilon=2\cdot 10^{-2}$. The line capacities have been chosen to be equal to $1.5$ times the average line power flow $\nu=W\mu$, and we used $q=10^{-4}$. The data for $\mu$, line susceptances and network topology have been extracted from the MATPOWER package~\cite{Zimmerman2011}. The regions $\Rbt_q$ and $\Rtt_q$ have been omitted since the calculations were intrinsically computationally unstable, as argued at the end of Section~\ref{sec2}.
Note that our capacity regions are indeed convex, and polyhedral.

\vspace{-0.1cm}

\begin{figure}[!h]
	\centering
	\subfloat[3-bus cycle network]{\includegraphics[scale=0.202]{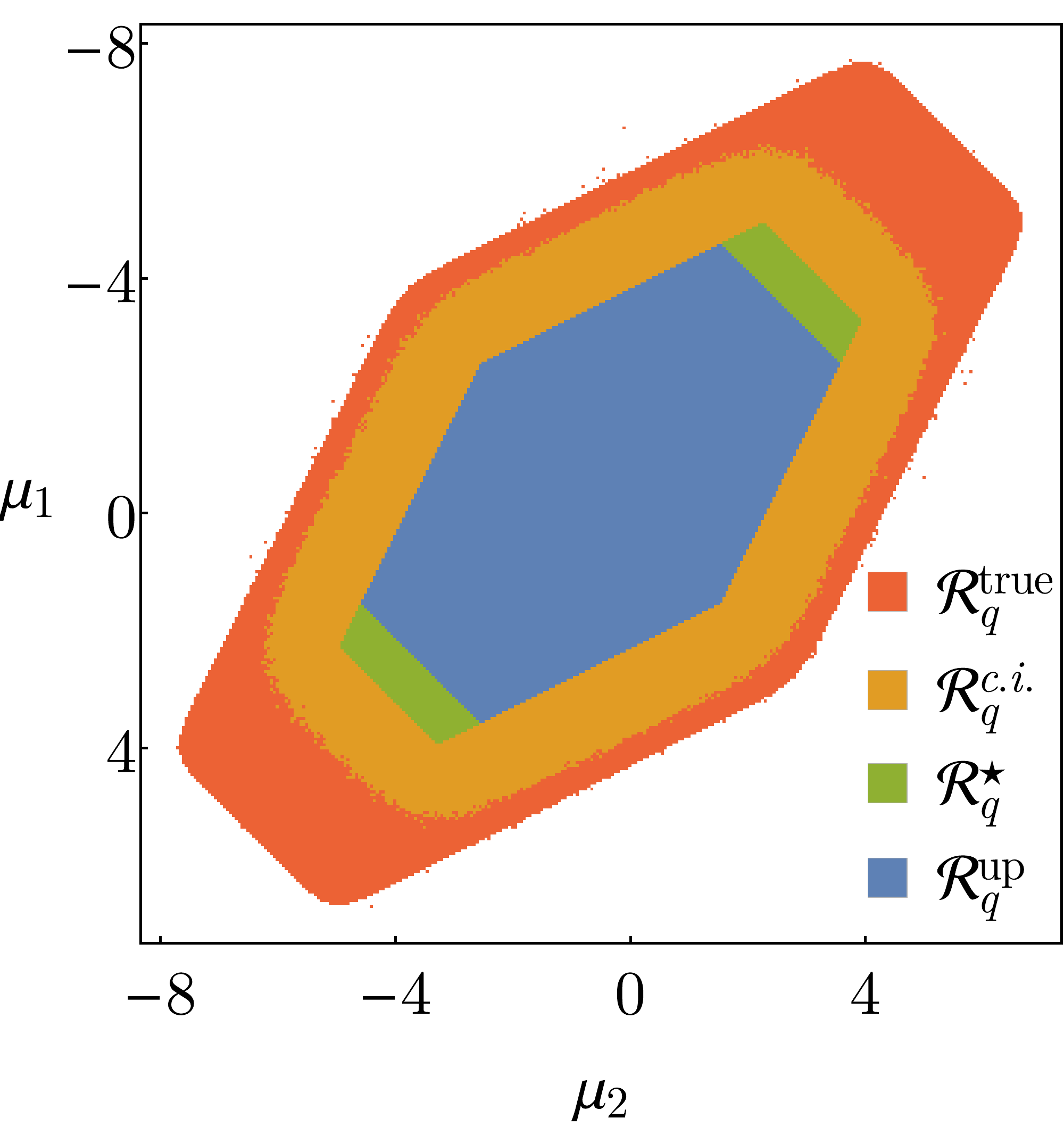} \label{fig:3bus}}
	\hspace{0.0cm}
	\subfloat[IEEE 14-bus network]{\includegraphics[scale=0.202]{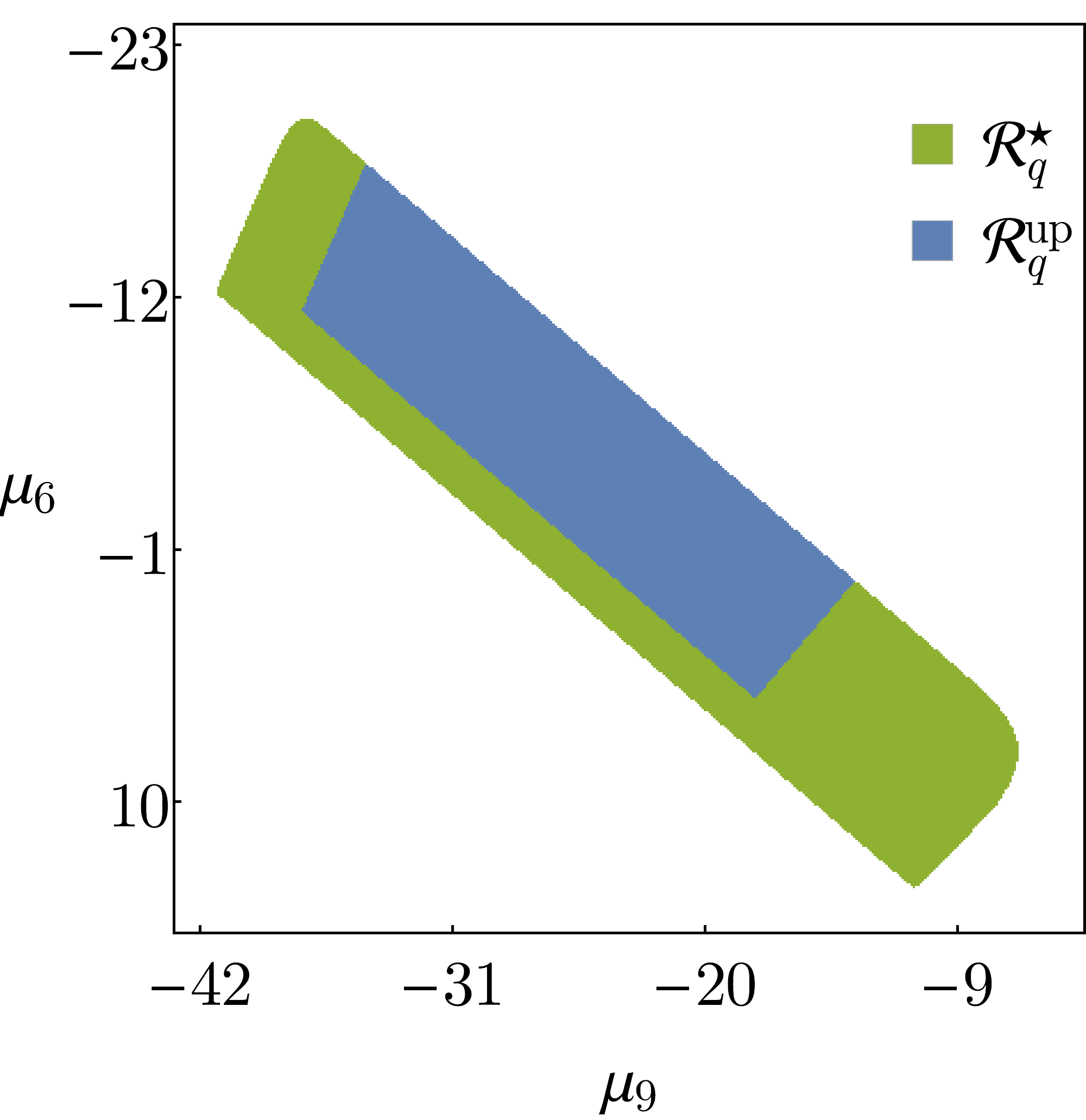} \label{fig:14bus}}
	\label{fig:regions}	
	\caption{Capacity regions comparison}
\end{figure}

\section{Mathematical tools}
\label{sec6}
\begin{prop}[Unilateral concentration inequality for the maximum of multivariate Gaussian random variables]\label{prop:borell}$ $\\
Let $X=(X_1,\ldots,X_k) \sim \N_{k}(\mu,\Sigma)$ be a multivariate Gaussian random variable, and let $\delta_i :=\sqrt{\Sigma_{i,i}}$ be the standard deviation of $X_i$, $i=1,\ldots,k$. The following concentration inequality holds for every $s\geq 0$:
\[
	\mathbb P \left ( \max_i |X_i| -\E \max_i |X_i| \geq s \right ) \leq 	\exp \left (-\frac{s^2}{2 \max_i \delta_i^2}\right).
\]
\end{prop}
\begin{proof}
The multivariate Gaussian vector $X$ can be seen as an affine transformation $X=\sqrt{\Sigma}Z+\mu$ of a standard Gaussian vector $Z\sim \N_{k}(0,I_k)$. Then we apply\cite[Theorem 2.4]{Wainwright2015} to the random vector $Z$ choosing the function $h:\R^k\to\R$ that maps $Z$ into $h(Z) :=\max_{i=1,\ldots,k} |(\sqrt{\Sigma})_i Z+\mu_i|$. A straightforward computation shows that $h$ is a Lipschitz function with Lipschitz constant equal to $\max_{i=1,\ldots,k} \delta_i$.
\end{proof}

\begin{proof}[Proof of Theorem~\ref{thm:concentration}]
Write
\[
	\P(\C) =\P \Big (\max_i|f_i|-\E \max_i|f_i|\geq 1-\E \max_i|f_i|\Big ).
\]
Set $s:=1-\E \max_i|f_i| >0$ and apply Proposition~\ref{prop:borell} to $f$. Inequality~\eqref{eq:bound_C} follows as  the standard deviation of $f_i$ is equal to $\sigma_i$, in view of definition~\eqref{eq:defsigmai}.
\end{proof}

\begin{lem}[Upper bounds for the risk level]\label{lem:upperbound_rmu}
Let $r(\mu):=\E \max_i|f_i|$, and define
\[
	r^\star(\mu):=\inf_{s\in(0,+\infty)} \left\{ \frac{\log(2m)}{s}+ \max_{i=1,\dots,m} \left ( \frac{\sigma_i^2}{2} s + |\nu_i| \right ) \right \}.
\]
Then
\begin{equation}
\label{eq:rmubounds}
	r(\mu) \leq r^\star(\mu)\leq \max_i |\nu_i| + \max_i \sigma_i \sqrt{2 \log(2m)}.
\end{equation}
\end{lem}
\begin{proof}
Take $2m$ random variables $Y_1,\dots,Y_{2m}$ defined as
\[
	Y_j:=
	\begin{cases}
		f_j & \text{ if } j=1,\dots,m,\\
		-f_{j-m} &\text{ if } j=m+1,\dots,2m.	
	\end{cases}
\]
From the definition of these random variables it immediately follows that $\max_{i=1,\dots,m} | f_i | = \max_{j=1,\dots,2m} Y_j$ and therefore $	\E \max_{i} | f_i | = \E \max_{j} Y_j$. Note that
\begin{equation}
	\lambda_j := \E Y_j=
	\begin{cases}
		\nu_j 			& \text{ if } j=1,\dots,m,\\
		-\nu_{j-m} 	&\text{ if } j=m+1,\dots,2m,		
	\end{cases}
\end{equation}
and $	\Var Y_j= \Var Y_{j+m} = \sigma_i^2$ for every $j =1, \dots,m$. For every $j=1,\dots,2m$ let $m_j(s):=\E \left ( \mathrm{e}^{s Y_j} \right ) = \mathrm{e}^{\sigma_j^2 s^2 /2 + \lambda_j s}$ be the moment generating function of the random variable $Y_j$. Following~\cite{Dasarathy2011}, for any $s \geq 0$ we have
\[
	\mathrm{e}^{s \, \E \max_j Y_j} \hspace{-0.075cm} \leq  \E ( \mathrm{e}^{s \max_j Y_j} ) = \hspace{-0.075cm} \sum_{j=1}^{2m} m_j(s) \leq 2m \max_j  \E ( \mathrm{e}^{s Y_j} ).
\]
Taking the $\log$ on both sides and rearranging we obtain
\begin{align*}
	&\E \max_j Y_j
	\leq \inf_{s \in (0, \infty)} \frac{1}{s} \log \left ( 2m \cdot \E \Big ( \max_{j=1,\dots,2m} \mathrm{e}^{s Y_j} \Big ) \right )\\
	& =\hspace{-0.075cm} \inf_{s \in (0, \infty)} \left \{ \frac{\log (2m) }{s} + \hspace{-0.075cm} \frac{1}{s} \log \left [ \max_{j=1,\dots,2m} \left (\mathrm{e}^{\sigma_j^2 s^2/2 + \lambda_j s} \right ) \right ] \right \},
\end{align*}
yielding the first bound, since the RHS is equal to $r^\star(\mu)$. If we now denote $\hat\nu := \max_{j=1,\dots,2m} \lambda_j = \max_i \nu_i$ and $\hat\sigma^2 = \max_i \sigma_i^2$, we have $m_j(s) \leq M(s)$ for all $s \geq 0$ and for every $j=1,\dots,2m$. Thus
\[
	\E \max_i Y_i \leq  \frac{\log (2m)}{s} + \frac{\hat \sigma^2}{2} s+ \hat{\nu}.
\]
Optimizing over $s$ in $(0,+\infty)$ and finding the optimum equals $s=\hat \sigma^{-1} \sqrt{2 \log(2m)}$,  we get
$
	\E \max_i Y_i \leq \hat{\nu} + \hat \sigma \sqrt{2 \log(2m)},
$
proving the other inequality in~\eqref{eq:rmubounds}.
\end{proof}

Lastly, we want to make some final remarks on how to calculate $r^\star(\mu)$ which is
 the infimum over $(0,\infty)$ of 
\[
	g(s):=\frac{\log(2m)}{s}+ \max_{i=1,\dots,m} \left  ( \frac{\sigma_i^2}{2} s + |\nu_i| \right ).
\]
This can be seen as the point-wise maximum of $m$ functions $g_k(s):=\frac{\log(2m)}{s}+  \frac{\sigma_k^2}{2} s + |\nu_k|$, $k=1,\dots, m$. Note that $r^\star(\mu)$ can be computed by evaluating of the function $g$ into at most $m + m(m-1)/2$ points and then take the minimum value: the candidate points are the $m$ local minima of the functions $g_1(s),\dots,g_m(s)$ (which are $s^\star_i := \sqrt{2 \log(2m)} /\sigma_i$, $i=1,\dots,m$), and the  points $s_{i,j} := 2(|\nu_i |- |\nu_j|)/(\sigma_j^2 - \sigma_i^2), \quad i,j=1,\dots,m, \, \, i \neq j,$
(if they exist and positive) of the lines $\frac{\sigma_i^2}{2} s + |\nu_i|$ and $\frac{\sigma_j^2}{2} s + |\nu_j|$ with , which are at most $m(m-1)/2$.
This analysis implies that the resulting capacity region is convex and polyhedral.

\section{Concluding remarks}
\label{sec5}
Probabilistic techniques, in particular powerful upper bounds for Gaussian random vectors, can be applied to generate explicit upper bounds for failure probabilities and corresponding safe capacity regions. The resulting regions are polyhedral, and can be characterized in such a way that they can be incorporated in optimization routines, such as OPF.
In an extended version of this paper we will show that our upper bounds give the correct asymptotic estimate of the failure probability in the small-noise large deviations regime as studied in \cite{Nesti2016}, i.e.\ our bounds are asymptotically sharp.
We will also extend the scope of our method as it is not limited to the assumptions in Section~\ref{sec2}: (i) the static analysis we consider can be extended to the dynamic situation as considered in~\cite{Nesti2016,Wadman2016}; (ii) the Gaussian assumption may be relaxed by the ideas in~\cite{Boucheron2013}; (iii)  other performance measures, like the probability that several lines fail, can be analyzed.

\bibliographystyle{IEEEtran}
\bibliography{library2}
\end{document}